\documentclass[11pt]{amsart}
\usepackage{graphics}
\usepackage{hyperref}

\newtheorem{theorem}{Theorem}[section]
\newtheorem{lem}[theorem]{Lemma}
\newtheorem{prop}[theorem]{Proposition}
\newtheorem{cor}[theorem]{Corollary}

\theoremstyle{definition}
\newtheorem{definition}[theorem]{Definition}

\theoremstyle{remark}

\numberwithin{equation}{section}

\newcommand{\Prob}{\text{Prob}}
\newcommand{\tPr}{\tilde{\Pr}}
\newcommand{\A}{\mathcal{V}}
\newcommand{\Ss}{\mathcal{S}}
\newcommand{\half}{\frac 12}
\newcommand{\K}{\mathfrak{K}}

\newcommand{\VT}{\mathsf{V}}
\newcommand{\X}{\mathcal{X}}

\begin{document}

\title{Clustering and expected Seat-Share for district maps}

\author {Kristopher Tapp}
\email{ktapp@sju.edu}
\begin{abstract}
In the context of modern sampling methods for redistricting, we define a natural measurement of the \emph{clustering} of a political party, and we study how clustering affects the expected election outcome.  We first prove general results and then apply this framework to understand how the political geography in Pennsylvania affects the expected outcome of congressional elections.
\end{abstract}
\maketitle

\section{Introduction}
The mathematical literature on redistricting and detecting gerrymandering has recently begun to coalesce around sampling methods.  The basic idea is to generate an ensemble of thousands or millions of randomly generated redistricting plans, and the underlying fairness principle is that an enacted plan should not be an outlier.  Among other things, this means that one political party's seat-share should not differ by too many standard deviations from its average seat-share among the ensemble.

Sampling methods have been applied (and have formed the basis of legal briefs) for Wisconsin~\cite{DUKE}, North Carolina~\cite{Gupta,DUKE2}, Pennsylvania~\cite{Duch} and other states.  The political geography of most of these states (which means the manner in which Democrat and Republican voters are spatially distributed across the state, as snapshotted by the vote tallies from a particular race) results in a structural advantage for Republicans.  The sampling methods allow this structural advantage to be separated from the advantage caused by an allegedly biased map.

This structural advantage is typically attributed to the way that Democrat voters are packed into dense cities while Republican voters are more evenly distributed across rural areas (see for example~\cite{Dontblamemaps}), but this is a bit vague and speculative.  The sampling methods themselves do not illuminate the reason for structural bias.  Other possible causes include the spacial distribution of non-voters and the state's redistricting criteria.  Precisely formulating the claim that ``Democrats are disadvantaged because they are packed into cities'' requires sampling methods that are much newer than the claim itself.

Is a political party advantaged or disadvantaged (with respect to its seat-share averaged over an ensemble) by being more clustered?  The authors of~\cite{Davis} studied this question for a toy model state with $25$ voters arranged in a $5$-by-$5$ grid to be partitioned into five contiguous districts.  Their ensemble contained \emph{all} 4006 possible district plans, and they defined ``clustering'' as the portion of voter-pairs who voted for the same party (among all rook-adjacent voter-pairs).  Among other results, they found a positive correlation between a minority party's clustering and its expected seat share.

New methods are needed to understand larger grids and/or real state data.  One purpose of this paper is to propose a clustering measurement that we call \emph{statistical clustering} because it depends, not just on the graph representing where each party's voters live, but also on the ensemble of maps.  The advantage of this dependence is that we are able to prove very general results relating a party's statistical clustering to its expected seat-share

The remainder of this paper is organized as follows.  In Section 2, we describe a general setup for modelling a state with two political parties under the simplifying assumption that everyone votes.  In Section 3, we define the \emph{statistical clustering} of a political party.  As with any reasonable measurement of clustering, this measurement is lowest when the party's voters are distributed homogeneously across the state, and is highest when they are all segregated into one geographically isolated region of the state.  We prove bounds on a party's average seat-share in terms of its clustering.  With very low clustering, our bounds virtually guarantee that a majority party will win all of the seats.  With very high clustering, our bounds guarantee an outcome close to proportionality.  These two limit results (for minimal and maximal clustering) are unsurprising.  What is new and perhaps useful for other applications is the abstract framework for addressing (in a general setting, not specific to any particular state or sampling method) the key question of how political geography affects the expected outcome of an election.

In Section 4, we apply these ideas to study Pennsylvania Congressional maps.  Single number measurements like statistical clustering don't seem to provide very strong bounds for an actual state.  This is because a party's performance has more to do with \emph{how} it's clustered than \emph{how much} its clustered.  But a key Lemma from Section 3 does make possible some precise statements about how the Democrat's performance relates to their concentration in cities.  With respect to 2016 Presidential voting data (in which the Democrat two-party vote share was about $50\%$), and with respect to a specific choice of ensemble (described in Section 4), the Democrat's expected seat share is $37\%$, which breaks down as:
$$38\% = 29\% \text{(\emph{from Philadelphia})} + 7\% \text{(\emph{from Pittsburgh})}+1\% \text{(\emph{from all else})}.$$
It is perhaps not obvious that it makes sense to additively separate the contributions from different cities like this; the framework that makes this possible a discussed in Section 4.

Finally, the appendix contains a new general inequality for random variables.  Our bounds in Section 3 are based on this inequality, so readers interested in the proof details should begin with the appendix.

\section*{Acknowledgments} The author is pleased to thank Paul Klingsberg for valuable feedback.

\section{setup}
It is common to use a graph $G$ to model a state that must be divided into $k$ districts.  The vertices of $G$ represent the smallest units of population out of which districts are to be formed (for example, voter tabulation districts, precincts or wards).  Let $\VT$ denote the set of vertices of $G$.  Two vertices are connected by an edge if the corresponding pair of geographical regions share a boundary of non-zero length.

The vertices are weighted by a population function $p$; more precisely, if $v\in \VT$, then $p(v)\in(0,1)$ represents the fraction of the state's population that resides in $v$.  We assume that everyone votes for either party $A$ or party $B$, so the population of each vertex $v\in \VT$ splits up correspondingly as:
$$p(v) = p_A(v) + p_B(v),$$
where $p_A(v)$ (respectively $p_B(v)$) denotes the fraction of the state's population that both resides in $v$ and votes for party $A$ (respectively party $B$).
The functions $\{p,p_A,p_B\}$ on $\VT$ can be applied to any subset $C\subset \VT$ in the obvious additive way; for example, $p(C)=\sum_{v\in C}p(v)$, and similarly for $p_A$ and $p_B$.

A \emph{district} means a subset $D\subset \VT$ such that $p(D) = 1/k$ (this is a minor simplification, since typically the district population is only required to \emph{approximately} equal $1/k$).  A \emph{map} means a partition of $\VT$ into $k$ districts.

Our starting point is a probability function denoted ``$\Pr$'' on the space of maps.  In this article, we allow $\Pr$ to be arbitrary, but in practice one would have $\Pr(M)=0$ for any map $M$ that violates any of the state-specific redistricting rules like district contiguity, district compactness, and compliance with the Voters Rights Act.  For example, $\Pr$ might be uniform among compliant maps, or might be weighted against barely-compliant maps.  We refer to~\cite[page 7]{DUKE2} for an overview of the varying methods that have been used to create ensembles of maps, including constructive randomized algorithms, optimization algorithms, and MCMC algorithms.  In principle, any ensemble-generation method induces probability function on the space of maps, but the MCMC approach matches best with the viewpoint of this paper because it generates an ensemble that can theoretically be proven to be a random sample from the space of all possible maps with respect to an explicitly described probability function.

Note that $\Pr$ induces a probability function $\tilde{\Pr}$ on the space of districts, defined such that for any district $D$,
$$\tPr(D) = \frac{1}{k}\cdot\sum_{M\ni D}\Pr(M),$$
where the sum is over all maps $M$ that contain $D$ as one of their districts.  Notice that choosing a random district is equivalent to the following two-step process: first choose a random map using the $\Pr$ function, then choose one of its $k$ districts uniformly at random.

\section{Statistical Clustering}

In this section, we propose a natural way to measure the clustering of the members the political parties.  For this, first consider the \emph{vote-share} random variable, $\A$, defined such that for each district $D$,
\begin{equation}\label{E:Vv}\A(D) = \frac{p_A(D)}{p(D)} = k\cdot p_A(D),\end{equation}
which is just the fraction of the district that is loyal to party $A$.  Notice that $\A$ is a random variable with respect to $\tPr$.  Denote $\mu_{\A}=E(\A)$ and $\sigma^2_{\A}=\text{Var}(\A)$.  The following lemma says that $\mu_{\A}$ equals Party $A$'s statewide vote-share.

\begin{lem}\label{L:count} $\mu_{\A} = p_A(\VT)$.
\end{lem}
This conclusion is very natural.  For example, if $35\%$ of the state's population is loyal to party $A$, the lemma says that on average $35\%$ of a random district is loyal to party $A$.
\begin{proof}
\begin{align*}
\mu_{\A}  & = \sum_D\tPr(D)\cdot \A(D) = \sum_D\left(\frac 1k \sum_{M\ni D}\Pr(M)\right)\cdot \A(D) \\
       & = \sum_M\Pr(M)\cdot\left(\frac 1k\sum_{D\in M} \A(D)\right)
         = \sum_M\Pr(M)\cdot\left(\sum_{D\in M} p_A(D)\right)\\
         & = \sum_M\Pr(M)\cdot p_A(\VT) = 1\cdot p_A(\VT).
\end{align*}
\end{proof}
\begin{definition} The \emph{statistical clustering} is:
$\K = \frac{\sigma_\A^2}{\mu_\A(1-\mu_\A)}.$
\end{definition}
The value $\K$ is defined here with respect to Party $A$, but it is straightforward to verify that the measurement would be unchanged if the roles of $A$ and $B$ were swapped throughout. In other words, the statistical clustering of Party $A$ equals the statistical clustering of Party $B$, which is why the non-party-specific terminology is appropriate.

According to the appendix, $\K\in[0,1]$ measures the variance of $\A$ relative to the maximal variance possible for the given value of $\mu_\A$, and this maximum occurs when the support of $\A$ is $\{0,1\}$.  Values of $\K$ close to zero occur in the homogeneous situation when all vertices (and hence all districts) have about the same fraction of party $A$ voters.  On the other hand, values of $\K$ close to $1$ occur when a randomly chosen district is very likely to either be filled entirely with party $A$ voters or contain no party $A$ voters.  In natural applications, we expect this to happen when party $A$ voters are segregated into a geographically isolated region of the state.

The term ``geographically isolated'' requires qualification here.  Our definition of $\K$ depends not only on the graph $G$ but also on the probability function $\Pr$.  Because of this dependence, ``clustering'' is not the right word for what $\K$ measures in certain unnatural mathematically-contrived examples, such as when $\Pr$ is uniform among all maps with no regard for contiguity or other geographical redistricting criteria.  But in natural applications, such as when MCMC methods are applied to a particular state with a reasonably large number of districts, we propose that ``clustering'' is exactly what $\K$ measures because of the way that $\Pr$ encodes the geography of the state.  Furthermore, unlike graph-theoretic definitions for clustering that depend only on $\{G,p_A,p_B\}$, our statistics-based measurement $\K$ can be easily related to the expectation for the outcome of an election, which is the goal of the remainder of this section.

To this end, we consider another random variable.  For any fixed map $M$, let $\Ss(M)$ denote the \emph{seat-share} of party $A$, which is defined as $\frac 1k$ times the number of districts $D$ of $M$ in which $p_A(D)>p_B(D)$.  Notice that $\Ss$ is a random variable with respect to $\Pr$.  Let $\mu_{\Ss}$ denote its expected value, which we will now prove is equal to the probability that party A wins a randomly chosen district:
\begin{lem}\label{L:Seatshare}
$\mu_{\Ss} = \Prob\left(\A>\half\right)$.
\end{lem}
\begin{proof}
For any district $D$, define
$\delta(D)=\begin{cases}1 & \text{if }p_A(D)>p_B(D) \\ 0 & \text{otherwise}  \end{cases}.$
\newline Note that $p_A(D)>p_B(D)$ if and only if $\A(D)>\frac 12$, so:
\begin{align*}
\Prob\left(\A>\half\right)  & = \sum_D\tPr(D)\cdot \delta(D) = \sum_D\left(\frac 1k \sum_{M\ni D}\Pr(M)\right)\cdot \delta(D) \\
       & = \sum_M\Pr(M)\cdot\left(\frac 1k\sum_{D\in M} \delta(D)\right) \\
       &  = \sum_M\Pr(M)\cdot \Ss(M) = \mu_{\Ss}.
\end{align*}
\end{proof}

Combining Lemma~\ref{L:Seatshare} with the $\theta=\half$ case of Corollary~\ref{C:BD} from the appendix yields the main result of this section:

\begin{prop}\label{P:A}\hspace{1in}\begin{enumerate}
\item $\mu_\Ss \geq \left(2-2\K\right)\mu_\A^2 - \left(1-2\K\right)\mu_\A$.
\item If $\mu_\A>\half$, then $\mu_\Ss \geq 1-\K\cdot \frac{\mu_\A(1-\mu_\A)}{\left(\mu_\A-\half\right)^2}$.
\end{enumerate}
\end{prop}

Proposition~\ref{P:A} provides a lower bound on Party $A$'s expected seat-share.  If one exchange the roles of $A$ and $B$ throughout, then the proposition instead gives a lower bound on Party $B$'s expected seat-share, which is equivalent to an \emph{upper} bound on Party $A$'s expected seat-share\footnote{This requires the assumption that tied districts are negligibly rare.}.  Thus, after applying the proposition to both parties, one obtain lower and upper bounds on $\mu_\Ss$.

The remainder of this section is devoted to explaining and interpreting Figure~\ref{F:Seatshare}, which visualizes these upper and lower bounds.  To provide context, we first review the following well-known general bound on $\mu_\Ss$.  Party $A$'s seat-share for \emph{any} map $M$ is bounded in terms of party $A$'s vote-share $\mu_\A$ as follows: $2\mu_\A-1\leq \Ss(M)\leq 2\mu_\A$; see~\cite[Lemma 1]{Tapp}.  In particular,
\begin{equation} 2\mu_\A-1\leq \mu_\Ss\leq 2\mu_\A.\label{E:purple}\end{equation}

In Figure~\ref{F:Seatshare}, the black lines represent Equation~\ref{E:purple}, the red curves represent part (1) of Proposition~\ref{P:A} applied to both parties, while the blue curves represent part (2).  The grey shaded regions represent the ``feasible'' points ($\mu_\A,\mu_S)$ for the given values of $\K$, where ``feasible'' means consistent with all of the inequalities under consideration.

\begin{figure}[ht!]\centering
   \scalebox{.20}{\includegraphics{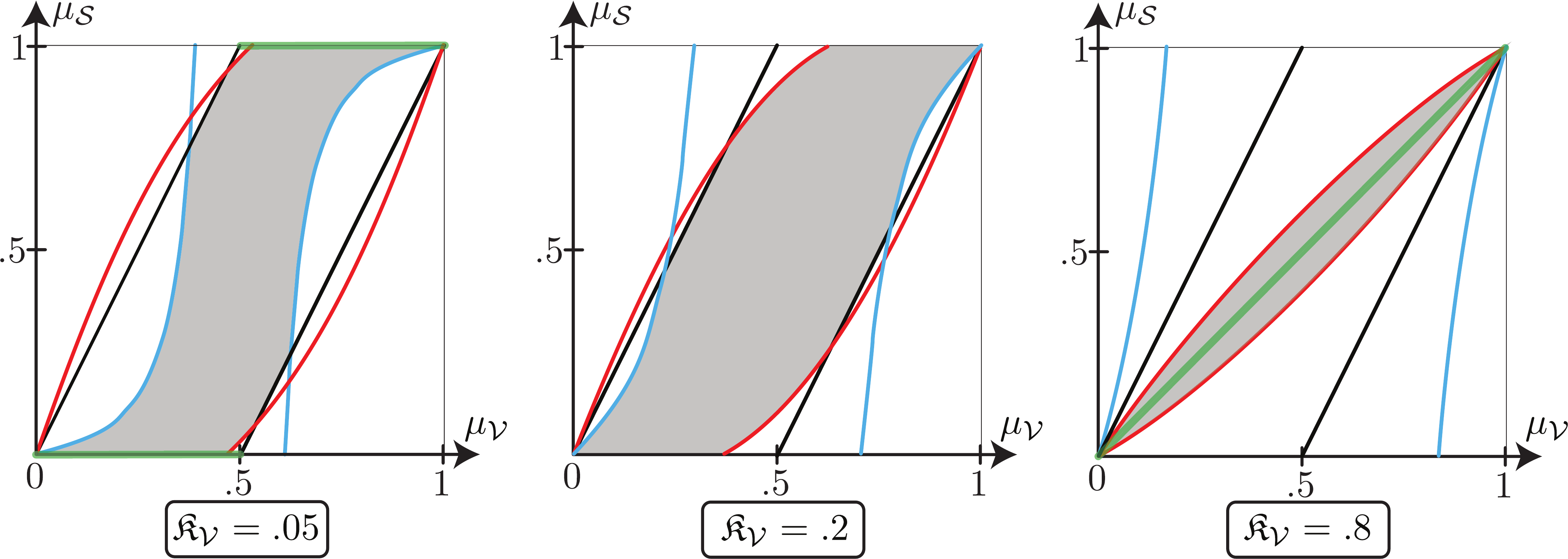}}
\caption{Bounds on Party $A$'s expected seat-share $\mu_\Ss$ as a function of Party $A$'s vote-share $\mu_{\A}$ (for three values of $\K$).}\label{F:Seatshare}
   \end{figure}

For small values of $\K$, the blue (Tchebysheff) curves dominate the story and ensure that a party with a reasonably large majority can expect to win all of the seats.  In fact, in the limit as $\K\rightarrow 0$, the area of the grey region in Figure~\ref{F:Seatshare} shrinks to zero and its shape converges to the ``majority takes all'' curve shown green in the left graph of the figure.

For large values of $\K$, the red curves dominate the story and guarantee an outcome close to proportionality.  In fact, in the limit as $\K\rightarrow 1$, the area of the grey region in Figure~\ref{F:Seatshare} shrinks to zero and its shape converges to the proportionality line $\mu_\Ss = \mu_\A$ shown green in the right graph of the figure.

In summary, a majority party will favor $\K\approx 0$ (winning all the seats), while a minority party will favor $\K\approx 1$ (achieving proportionality).  But in the next section we'll see that when $\K$ isn't close to $0$ or $1$, a party's expected performance depends less on $\K$ and more on the manner in which the party is clustered.

\section{The political geography of Pennsylvania}
In this section, we use the framework developed in the previous section to empirically understand the political geography of Pennsylvania and its effect on the expected outcome of congressional elections.  All empirical results from this section are derived in a Jupyter notebook available at~\cite{Tapp_site}

Single number measurements like statistical clustering only provide weak bounds in Pennsylvania because a party's performance has more to do with \emph{how} it's clustered than \emph{how much} its clustered.  Nevertheless, we will demonstrate that Lemma~\ref{L:Seatshare} is useful tool for making precise statements about how the Democrat party's performance relates to its concentration in the greater Philadelphia and Pittsburgh regions.

All of the results in this section are based on the Pennsylvania data from~\cite{Pa_shape} made available by the \emph{Metric Geometry and Gerrymandering Group}.  Specifically, we use MGGG's precinct shape files and partisan voting information from the 2016 presidential election results.  We let $G$ denote the graph of Pennsylvania as described in the previous sections, which has $|\VT|=9255$ precincts.

We build an ensemble of 10000 maps generated by the Recombination algorithm introduced in~\cite{Recom} and made available as open source code at~\cite{Recom_alg}.  Our initial seed is the Remedial map (the map used in the 2018 congressional election).  Our only constrains are that the districts be equipopulus (up to a $2\%$ error) and compact (in the sense that the number of cut edges is less than $1.5$ times than that of the Remedial map).  In particular, we do not impose any constraints related to the Voters Rights Act or to the number of county splits.

Any method of generating an ensemble induces a probability function $\Pr$ on the set of all maps.   In our case, $\Pr$ is not explicitly understood, since one must overlay the Metropolis-Hasting algrorithm on top of the Recombination algorithm (as described in~\cite{Mergesplit}) in order for the ensemble to be a sample from an explicitly understood probability function like the uniform distribution on the set of maps satisfying the constraints.  In other words, we're choosing a simple fast algorithm at the cost of settling for a probability function that's not explicitly described.

From each of the 10000 maps in the ensemble, we uniformly randomly choose one of its $18$ congressional districts, so we have an ensemble of $10000$ \emph{districts}, which induces a probability function $\tilde{\Pr}$ on the set of all districts.

All averages with respect to $\Pr$ (respectively $\tilde{\Pr}$) will be calculated as sample averages over our ensemble or maps (respectively our ensemble of districts).

Let Party $A$ be the Democratic party and $B$ the Republican party.  Since Pennsylvania has non-voters and third-party voters, the definition of the ``two-party vote-share'' random variable $\A$ in Equation~\ref{E:Vv} must be modified as follows:
$$\A(D) = \frac{p_A(D)}{p_A(D)+p_B(D)}.$$

Lemma~\ref{L:count} does not generalize to the setting with non-voters, but empirically the following values only differ by about one percent:
$$ \mu_\A \approx .487,\qquad \frac{p_A(\VT)}{p_A(\VT)+p_B(\VT)}\approx .497.$$
Even if everyone voted, there would still be small difference between these values due to the use of sample averages.  Thus, the Democrat's statewide two-party vote share was just under $.5$ and was well-approximated by $\mu_\A$.

Lemma~\ref{L:Seatshare} remains valid even in the presence of non-voters, with essentially the same proof, so
\begin{equation}\label{E:peter}\mu_{\Ss} = \Prob\left(\A>\half\right),\end{equation}
with the ``seat-share'' random variable $\Ss$ defined as in the previous section.

The statistical clustering is:
$$\K = \frac{\sigma_\A^2}{\mu_\A(1-\mu_\A)}\approx \frac{(.16)^2}{.49(1-.49)}\approx .12.$$
Proposition~\ref{P:A} remains valid even in the presence of non-voters, but yields only the following weak bound:
$$\mu_\Ss \geq \left(2-2\K\right)\mu_\A^2 - \left(1-2\K\right)\mu_\A\approx .05.$$
Thus, Proposition~\ref{P:A} only guarantees the Democrats an expected seat-share of at least $5\%$.

The actual value is: $\mu_\Ss \approx .37,$ which according to Equation~\ref{E:peter} can be visualized as the area to the right of $.5$ in Figure~\ref{F:D_Density}

\begin{figure}[ht!]\centering
   \scalebox{.60}{\includegraphics{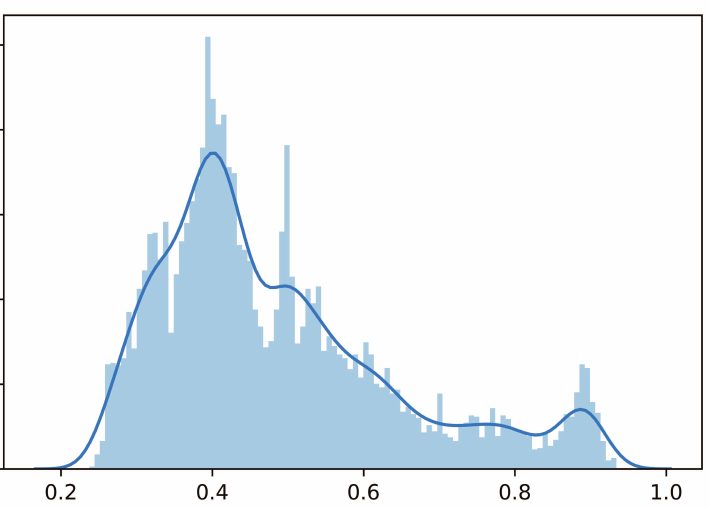}}
\caption{A density plot for $\A$.}\label{F:D_Density}
   \end{figure}

Even though the statistical clustering failed to account for the value $\mu_\Ss\approx .37$, in the remainder of this section we demonstrate that Equation~\ref{E:peter} allows us to make precise claims about how the Democrat's concentration in Philadelphia and Pittsburgh almost fully accounts for this value.

For this, we must first identify the Democrat strongholds in the Philadelphia and Pittsburgh regions.  First define
\begin{equation}\label{E:VAdef} \VT_A = \left\{v\in\VT\mid \frac{p_A(v)}{p_A(v)+p_B(v)}>.55\right\},\end{equation}
which is the set of precincts where the Democrats have at least $55\%$ of the two-party vote. The subgraph of $\VT$ induced by $\VT_A$ has $124$ connected components.  Let $\VT_1$ and $\VT_2$ denote the largest two of these components (ranked by total voting population).  For $i\in\{1,2\}$, let $\overline{\VT}_i\subset\VT$ denote the set of vertices that intersect the convex hull of $\VT_i$.  The boundaries of $\overline{\VT}_1$ and $\overline{\VT}_2$ are shown in Figure~\ref{F:Pit_Phl} superimpose on a choropleth map in with each precinct $v\in\VT$ is colored on a red-to-blue scale according to the value $\frac{p_A(v)}{p_A(v)+p_B(v)}$.  For expository convenience, we will henceforth refer to $\overline{\VT}_1$ and $\overline{\VT}_2$ as \emph{Philadelphia} and \emph{Pittsburgh} respectively, even though they don't match the municipal boundaries.

\begin{figure}[ht!]\centering
   \scalebox{.75}{\includegraphics{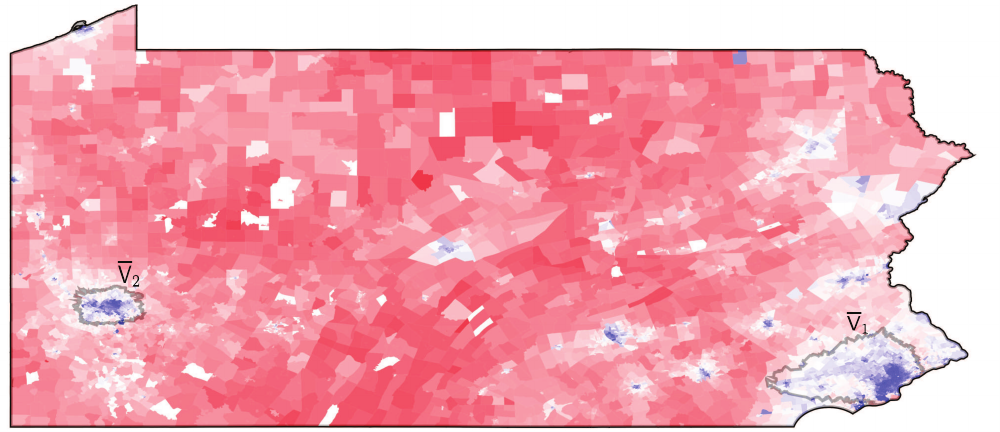}}
\caption{The two largest Democrat strongholds.}\label{F:Pit_Phl}
   \end{figure}

To understand the influence of these two strongholds, we define the following for each $i\in\{1,2\}$ and each district $D$:
$$\X_i(D) = \frac{p(D\cap \overline{\VT}_i)}{p(D)}$$
Notice that $\X_1$ and $\X_2$ are random variables with respect to $\tilde{\Pr}$ that encode the amount of overlap that districts have with $\overline{\VT}_1$ and $\overline{\VT}_2$.

Empirically,
$$\Prob(\X_1>0 \text{ and } \X_2>0) = 0,$$ that is, districts never overlap both Pittsburgh and Philadelphia.  Furthermore,
$$\Prob(\X_1=0 \text{ and } \X_2=0 \text{ and } \A>.5) \approx 0.01,$$
that is, there is only a $1\%$ chance that a randomly chosen district will be won by Democrats but will not overlap Philadelphia or Pittsburgh.  This observation represents a precise sense in which Philadelphia and Pittsburgh are essentially the \emph{only} Democratic strongholds that influence the party's expected election performance.

Equation~\ref{E:peter} allows us to additively separate the influences of Philadelphia and Pittsburgh, as follows:
\begin{align}
\mu_{\Ss}
   & = \Prob\left(\A>\half\right) \notag\\
   & = \Prob(\X_1=0 \text{ and } \X_2=0 \text{ and } \A>.5)\notag\\
   &  \qquad + \Prob(\X_1>0)\cdot\Prob\left(\A>\half\mid \X_1>0\right) \notag\\
   & \qquad + \Prob(\X_2>0)\cdot\Prob\left(\A>\half\mid \X_2>0\right)\notag\\
   & \approx (.01) + (.35)(.83) + (.19)(.37) \label{E:graciegu}\\
   & \approx (.01) + (.29) + (.07) \approx .37\notag
\end{align}
In summary, Democrats expect on average to win $37\%$ of the districts in a random map, and we have additively separated this value into the contributions from Philadelphia ($29\%$), Pittsburgh ($7\%$) and all else ($1\%$).

It is worth asking whether any of the empirical values in Equation~\ref{E:graciegu} (namely $.35$, $.83$, $.19$ and $.37$) can be better understood in terms of simpler information.  In truth all of these values depend on the ensemble (that is they depend on $\tilde{\Pr}$).  But since our ensemble is geometrically reasonable, one might ask whether any of these values could be roughly predicted just from the graph $G$ and the partisan functions $p_A$ and $p_B$.

Not much can be said about the numbers $.35$ and $.19$ (the probability that a randomly chosen district touches Philadelphia and Pittsburgh respectively).  These values are very roughly related to the total populations and placements of these cities (Philadelphia in the corner and Pittsburgh in the interior), but not much else can be said, primarily because the process of choosing a random district isn't very well-approximate by any simple process that forgets the rest of the map.

Next consider the empirical value: $\Prob\left(\A>\half\mid \X_2>0\right) = .37$; that is, the Democrats have a $37\%$ chance of winning a random district that touches Pittsburgh.  Figure~\ref{F:Pit_joint} exhibits the approximately linear relationship between $\X_2$ and $\A$ among the districts in the ensemble with $\X_2>0$.  The regression line could have been predicted fairly accurately because its extreme values are known: its  maximum value of about $.68$ approximately equals the Democrat's two-party vote share within Pittsburgh, while its minimum value of about $.34$ approximately equals the Democrat's two-party vote share in the red ocean surrounding Pittsburgh.

\begin{figure}[ht!]\centering
   \scalebox{.45}{\includegraphics{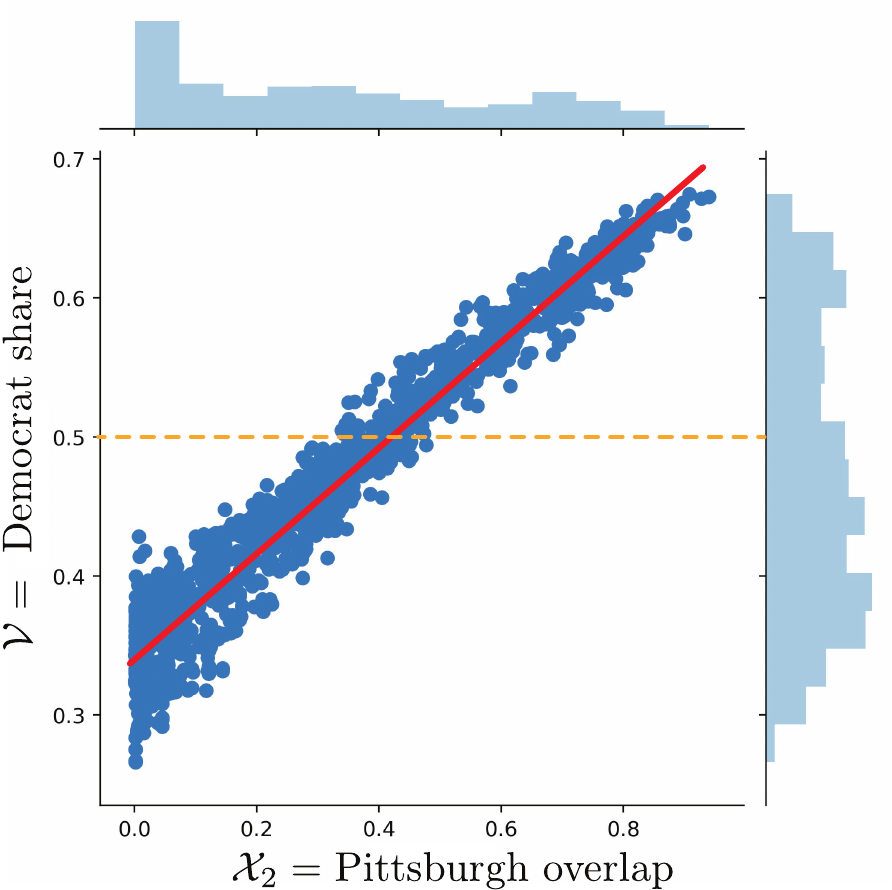}}
\caption{Joint plot for $\X_2$ and $\A$ among districts with $\X_2>0$.}\label{F:Pit_joint}
   \end{figure}

Next consider the empirical value: $\Prob\left(\A>\half\mid \X_1>0\right) = .83$; that is, the Democrats have an $83\%$ chance of winning a random district that touches Philadelphia.  Figure~\ref{F:Phl_joint} shows the correlation between $\X_1$ and $\A$ among the districts in the ensemble with $\X_1>0$.  The relationship is not as linear as in the Pittsburgh case.  The Choropleth of Philadelphia in Figure~\ref{F:Phl_blue} helps accounts for the nonlinearity.  The Democrat two-party vote-share is almost $100\%$ in Center City (near Philadelphia's east border) and decreases roughly linearly as one moves away from Center City.  This linear vote-share gradient would lead one to expect that $\A$ depends quadratically on $\X_1$, which roughly seems to be the case in Figure~\ref{F:Phl_joint}.

In summary, the framework of Section 3 allows us to additively separate the contributions of Philadelphia and Pittsburgh to the Democrat's expected seat-share.  The values of these contributions were computed empirically using the ensemble, but certain aspects of the calculations could have at least roughly been predicted from geometric and partisan information; that is, from $\{G,p_A,p_B\}$.

\begin{figure}[ht!]\centering
   \scalebox{.45}{\includegraphics{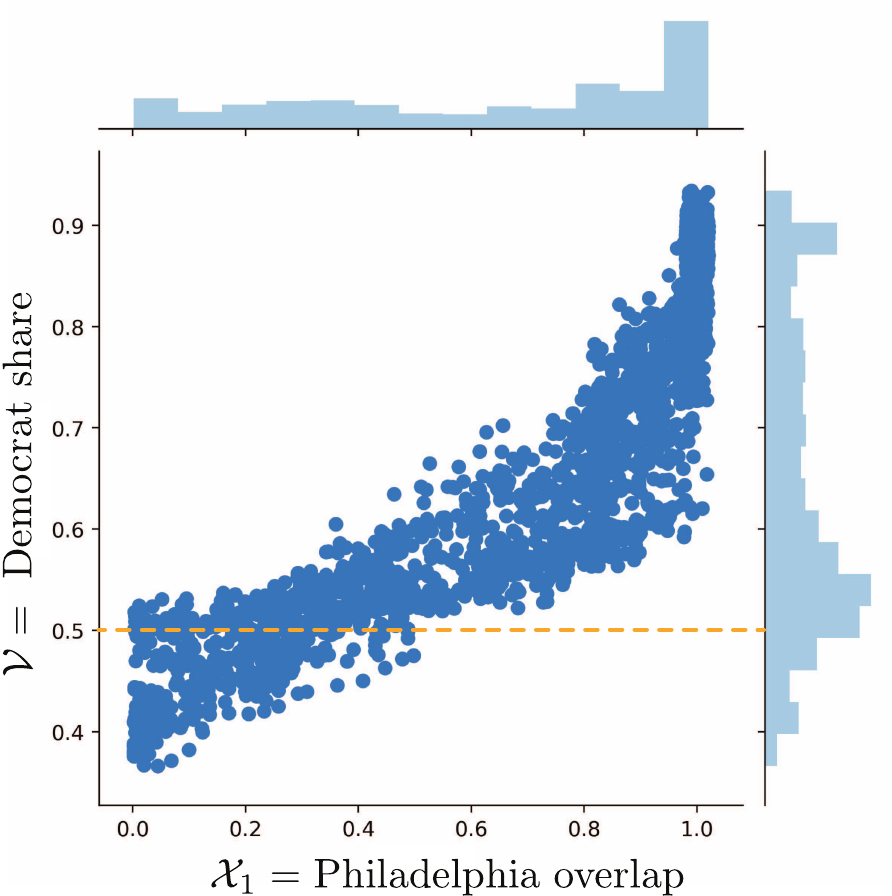}}
\caption{Joint plot for $\X_1$ and $\A$ among districts with $\X_1>0$.}\label{F:Phl_joint}
   \end{figure}

\begin{figure}[ht!]\centering
   \scalebox{.60}{\includegraphics{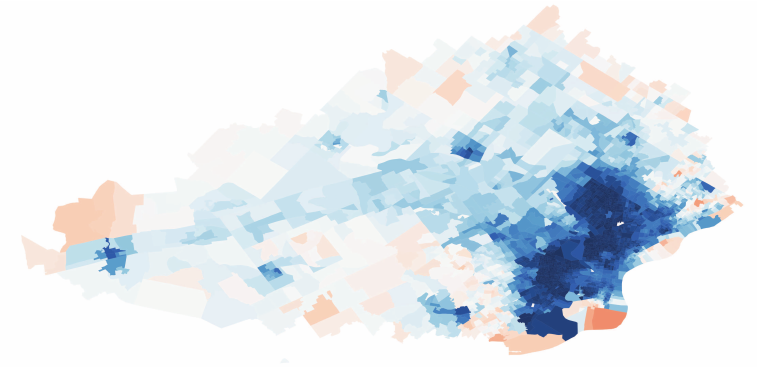}}
\caption{Choropleth of $\VT_1$ (Philadelphia).}\label{F:Phl_blue}
   \end{figure}

\section{APPENDIX: General Variance Bounds}
In this section, we prove a general result for a random variable with support in $[0,1]$.  The main result of Section 3 is based on this bound.

\begin{prop}\label{P:BD} Suppose that $X$ is a random variable with support in $[0,1]$, and $\theta\in(0,1)$ is a constant.  Denoting $\mu=E(X)$ and $\sigma^2=\text{Var}(X)$, we have:
\begin{itemize}
\item[] $\text{\emph{[Bhatia-Davis] }}\sigma^2\leq \mu(1-\mu).$
\end{itemize}
Furthermore:
\begin{enumerate}
\item $\Prob\left(X>\theta\right) \geq \frac{\sigma^2+\mu\left(\mu-\theta\right)}{1-\theta}$.
\item \emph{[Tchebysheff]} If $\mu>\theta$, then $\Prob\left(X>\theta\right) \geq 1- \left(\frac{\sigma}{\mu-\theta} \right)^2.$
\end{enumerate}
\end{prop}

The Bhatia-Davis inequality from ~\cite{BD} says more generally for a random variable with support in $[m,M]$ that $\sigma^2\leq(\mu-m)(M-\mu)$.  Part (2) of Proposition~\ref{P:BD} is a straightforward application of Tchebysheff's Inequality.  Part (1)  is related to the Bhatia-Davis inequality as follows: if $\sigma^2> \mu(\theta-\mu)$, then the Bhatia-Davis inequality implies that $X$ could not have support in $[0,\theta]$, and indeed this is exactly the cutoff after which part (1) concludes that $\Prob(X>\theta)>0$.

\begin{proof}[Proof of part (1) of Proposition~\ref{P:BD}] Define
$$F = \frac{\sigma^2+\mu\left(\mu-\theta\right)}{1-\theta}-\Prob\left(X>\theta\right).$$
We wish to prove that $F\leq 0$.  First observe that if the support of $X$ is in $\left\{0,\theta,1\right\}$, then the probability function for $X$ has the form of Table~\ref{T:data}, where $a=\Prob(X=1)=\Prob\left(X>\theta \right)$.  In this case, writing $\sigma^2$ in terms of $a$ and then solving for $a$ yields: $a=\frac{\mu^2 -\theta\mu + \sigma^2}{1-\theta}$, which means $F=0$.  So to prove that $F\leq 0$, we must demonstrate that this is the extreme case that maximizes $F$.

Assume for now that $X$ has finite support $\{a_1,...,a_n\}\subset[0,1]$, and let $p_i$ denote the probability of $a_i$, so $\mu=\sum p_i a_i$ and $\sigma^2 = \sum p_i(\mu-a_i)^2$. Notice that $\frac{\partial \mu}{\partial a_1} = p_1$.  Assume that $a_1\notin \left\{0,\theta,1\right\}$.  The following computation is similar to the first proof of~\cite{BD}:
\begin{align*}
\frac{\partial F}{\partial a_1}
 &= \frac{1}{(1-\theta)}\cdot \left(2p_1(a_1-\mu)-2p_1\sum p_i(a_i-\mu)\right) + \frac{p_1}{(1-\theta)}\cdot(2\mu-\theta) \\
 & = \frac{2p_1}{(1-\theta)}\left(a_1-\frac{\theta}{2}\right).
\end{align*}

There are three cases to consider.  If $a_1\in\left(0,\frac{\theta}{2}\right)$, then moving $a_1$ left towards $0$ increases $F$.  If $a_1\in\left(\frac{\theta}{2},\theta\right)$, then moving $a_1$ right towards $\theta$ increases $F$.  If $a_1\in\left(\theta,1\right)$, then moving $a_1$ right towards $1$ increases $F$.  From this, it follows that $F$ is maximized when the support of $X$ is $\left\{0,\theta,1\right\}$, as desired.

The case where $X$ does \emph{not} have finite support can be obtained by a standard limit approximation argument.
\end{proof}

\begin{table}
\begin{tabular}{ll}
\hline
$x$&\vline \,\,$\Prob(X=x)$   \\ \hline
$0$&\vline \,\,$1-a-\frac{1}{\theta}(\mu-a)$ \\
$\theta$&\vline \,\,$\frac{1}{\theta}(\mu-a)$ \\
$1$&\vline \,\,$a$           \\ \hline
\end{tabular}\caption{The extreme case: the support of $F$ is in $\left\{0,\theta,1\right\}$.}
\label{T:data}\end{table}

According to the Bhatia-Davis inequality, the value $\K=\frac{\sigma^2}{\mu(1-\mu)}\in(0,1)$ measures the size of the variance relative to the maximal variance possible for the given value of $\mu$.  Proposition~\ref{P:BD} can be easily rephrased in terms of $\K$ (rather than $\sigma$) as follows:
\begin{cor}\label{C:BD} Suppose that $X$ is a random variable with support in $[0,1]$, and $\theta\in(0,1)$ is a constant.  Denote $\mu=E(X)$ and $\sigma^2=\text{Var}(X)$.  Assume $\mu\notin\{0,1\}$ and define $\K=\frac{\sigma^2}{\mu(1-\mu)}\in(0,1)$.  Then
\begin{enumerate}
\item $\Prob\left(X>\theta\right) \geq \frac{1}{1-\theta}\left( \left(1-\K\right)\mu^2 + \left(\K-\theta\right)\mu\right)$.
\item \emph{(Tchebysheff)} If $\mu>\theta$, then $\Prob\left(X>\theta\right) \geq 1-\K\cdot \frac{\mu(1-\mu)}{\left(\mu-\theta\right)^2}$.
\end{enumerate}
\end{cor}


\bibliographystyle{amsplain}

\end{document}